\documentclass[letterpaper, 10 pt, conference]{ieeeconf}  

\IEEEoverridecommandlockouts                              
\usepackage{tabularx} 
\usepackage{graphicx} 
\usepackage{cite} 
\usepackage[final]{hyperref} 
\usepackage{amsmath} 
\usepackage{amssymb}  
\usepackage{mathtools, cuted,bm}
\usepackage{caption}
\usepackage{subcaption}
\usepackage{esvect}
\usepackage{lipsum, color}

\newcommand{\K}{\mathbf{K}}
\newcommand{\W}{\mathbf{W}}
\newcommand{\N}{\mathbf{N}}

\newcommand{\V}{\mathbf{V}}

\newcommand{\tjplus}{{\theta^{j}}^{+}}

\newcommand{\smallmat}[1]{\left[ \begin{smallmatrix}#1 \end{smallmatrix} \right]}

\hypersetup{
	colorlinks=true,       
	linkcolor=black,        
	citecolor=black,        
	filecolor=magenta,     
	urlcolor=black         
}

\newtheorem{theorem}{Theorem}
\newtheorem{lemma}[theorem]{Lemma}
\newtheorem{remark}{Remark}
\newtheorem{problem}{Problem}
\newtheorem{proposition}{Proposition}

\title{\LARGE \bf
	Direct Data-Driven Computation of Polytopic Robust Control Invariant Sets and State-Feedback Controllers
}

\author{Manas Mejari, Ankit Gupta
\thanks{This work has been accepted for publication, to appear at the 62nd IEEE Conference on Decision and Control (CDC 2023), Singapore.}
\thanks{M. Mejari is with IDSIA Dalle Molle
	Institute for Artificial Intelligence, Via la Santa 1, CH-6962 Lugano-Viganello, Switzerland. {\tt\small manas.mejari@supsi.ch}}
\thanks{A. Gupta is working with vehicle motion planning at Zenseact AB, Gothenburg, Sweden.
        {\tt\small ankit.gupta@zenseact.com}}%
\thanks{The activities of M. Mejari have been supported by HASLER STIFTUNG under the project \emph{INHALE: Interpretable Neural networks for Hybrid dynamicAL systEms.}
}
}

\begin{document}
\newcounter{tempEquationCounter}
\newcounter{thisEquationNumber}
\newenvironment{floatEq}
{\setcounter{thisEquationNumber}{\value{equation}}\addtocounter{equation}{1}
\begin{figure*}[!t]
\normalsize\setcounter{tempEquationCounter}{\value{equation}}
\setcounter{equation}{\value{thisEquationNumber}}
}
{\setcounter{equation}{\value{tempEquationCounter}}
\hrulefill\vspace*{4pt}
\end{figure*}

}
\newenvironment{floatEq2}
{\setcounter{thisEquationNumber}{\value{equation}}\addtocounter{equation}{1}
\begin{figure*}[!t]
\normalsize\setcounter{tempEquationCounter}{\value{equation}}
\setcounter{equation}{\value{thisEquationNumber}}
}
{\setcounter{equation}{\value{tempEquationCounter}}
\end{figure*}

}

\maketitle
\thispagestyle{empty}
\pagestyle{empty}

\begin{abstract}
This paper presents a direct data-driven approach for computing \emph{robust control invariant} (RCI) sets and their associated state-feedback control laws for linear time-invariant systems affected by bounded disturbances. The proposed method utilizes a single state-input trajectory generated from the  system, to compute a polytopic RCI set with a desired complexity and an invariance-inducing feedback controller, without the need to identify a model of the system. The problem is formulated in terms of a set of sufficient linear matrix inequality  conditions that are then combined in a semi-definite program to maximize the volume of the RCI set while respecting the state and input constraints.
We demonstrate through a numerical case study that the proposed data-driven approach can generate RCI sets that are of comparable size to those obtained by a model-based method in which exact knowledge of the system matrices is assumed.

\end{abstract}

\section{Introduction}\label{sec:introduction}
Ensuring safety is of paramount importance in the operation of feedback-controlled systems in various safety-critical applications such as autonomous driving and aircraft flight control. This can be achieved by guaranteeing that the system respects  state and input safety constraints at all times. Essentially, this requires imposing \emph{invariance} of a set, \emph{i.e.}, the system's states  when initialized within the set will never leave it. Therefore, set invariance theory has received a significant attention over the past few years,  particularly  for constrained systems and stability analysis~\cite{fb99,fb}. 

A set is called \emph{robust control invariant} (RCI) if from all  initial states within the set, an admissible control input exists, which keeps the state trajectories within the set for all bounded disturbances acting on the system~\cite{fb99}. Several contributions have been proposed in the literature to compute RCI sets and its associated controllers given a model of the system,  see for \emph{e.g.}, \cite{gmfp23auto,sr10,mf10, lc15, Tahir2015,tb10}. These are \emph{model-based} methods where the main underlying assumption is that a model of the true system is available. However, there are several challenges to obtain an accurate model of the system~\cite{hou13}. Computing a model from the data requires an additional system identification step, and the identified model may be inaccurate when only a few data samples are available, resulting in a large identification error. An inaccurate model can lead to loss of the invariance property as well as violation of constraints when operating in the closed-loop due to model-mismatch, as demonstrated in~\cite{zhong22}. On the other hand, even if an accurate model is available using first principles, it may be too complex for efficient controller synthesis and RCI set computation.

To overcome the limitations of model-based methods, recent developments have emphasized  data-driven approaches. One such class of methods falls under the category  of \emph{control-oriented identification}, which considers control design and invariance set computation together with model identification. It is shown that concurrent model selection with RCI set computation results in reduced conservatism~\cite{sam22,chen22}. Alternatively, the second category encompasses \emph{direct} data-driven control approaches~\cite{berb20, berb22, bis20, bisoffi23}, which synthesize robust controllers directly from the open-loop data of the system without the need for model identification. In this paper, we focus on the latter approach.

The direct data-driven methods presented in  \cite{bis20,bisoffi23}, compute a state-feedback controller from open-loop data to induce robust invariance in a given polyhedral set. However, these methods require that the set is fixed and known a priori. 
 A recent work \cite{zhong22} offers a  method for computing an invariant set as well as its associated feedback controller. While the  approach proposed in \cite{zhong22} does not require prior knowledge of the set, it constructs an \emph{ellipsoidal} invariant set. It should be noted that ellipsoidal sets are potentially more conservative than \emph{polyhedral} sets, as the latter present several theoretical and practical advantages over the ellipsoids via flexible and arbitrarily complex representation~\cite{fb}, with the only drawback of scalability.  

In this paper, we develop a direct data-driven approach to compute polytopic RCI sets and state-feedback controllers for unknown linear systems subject to bounded disturbances. Our method aims to address the challenge of constructing polytopic RCI sets of desired complexity directly from data. We utilize a single state-input trajectory generated in open-loop and derive data-based sufficient LMI conditions which can guarantee invariance and constraint satisfaction. 
The sufficient LMI conditions are then combined in an SDP program to maximize the volume of the RCI set. 
The proposed approach  can be seen as a data-driven counterpart to the model-based methods presented in~\cite{gkf17}. Specifically, we consider a state transformation in which the candidate RCI set is mapped into a known polytope as in~\cite{gkf17}. The state-input constraints in the transformed space are found to be affine inequalities. Furthermore, a novel data-based condition for invariance is derived using full block S-procedure. The advantage of having flexibility to choose the representational complexity of the RCI set is demonstrated using a numerical example. We point out that in~\cite{gkf17}, it is assumed that the exact model of the true system is known, while the approach presented in this paper neither requires knowledge of the model nor an additional system identification step.

Paper organization: The notation and preliminary results used in the paper are given in Section~\ref{sec:notation}. The problem of computing the RCI set and the controller is formalized in Section~\ref{sec:prob}. The proposed data-based conditions for invariance and constraints are described in Section~\ref{sec:LMI}. A one-step and an iterative algorithm based on these conditions is given in Section~\ref{sec:SDP} to obtain desirably large volume RCI set. Finally, in Section~\ref{sec:example}, the effectiveness of the proposed algorithm is demonstrated with a numerical example. 

\section{Notations and Preliminaries}\label{sec:notation}
The set of real $m \times n$ matrix is denoted by $\mathbb{R}^{m \times n}$ and $\mathbb{D}^n_+\in\mathbb{R}^{n\times n}$  denotes  the set of all diagonal  matrices with positive diagonal entries. An identity matrix of  dimension $n$ is denoted by $I_n$ and $e_i$ represent  and its $i$-th column. A matrix of zeros with appropriate dimension is denoted as $\bm{0}$.
The vector of ones with dimension $m$ is  denoted by $\bm{1}_{m}$.  
 $X \succ 0\,(\succeq 0)$ denotes a
positive (semi) definite matrix $X$. For compactness, in the text $*$'s will represent matrix entries that are uniquely identifiable from symmetry. Let $A\in \mathbb{R}^{m \times n}$ be a matrix written according to its $n$ column vectors as  $A = \smallmat{a_1 \ \cdots \ a_n}$, we define vectorization of $A$ as $\Vec{A} \triangleq \smallmat{a^{\top}_{1} \ \cdots \ a^{\top}_{n}}^{\top}$, which returns a vector of dimension $(mn \times 1)$, stacking the columns of $A$. For a finite set $\Theta_v = \{\theta^{1}, \theta^{2}, \ldots, \theta^{r} \}$ with $\theta^{j} \in \mathbb{R}^{n}$ for $i=1,\ldots, r$, the convex-hull of $\Theta_v$  is given by,
$\mathsf{conv}(\Theta_v) \triangleq \left\{ \theta\in  \mathbb{R}^{n}: \theta = \sum_{j=1}^{r} \alpha_j \theta^{j}, \mathrm{s.t} \ \sum_{j=1}^{r} \alpha_j =1, \alpha_j \in [0,1]  \right\} $.
For matrices $A$ and $B$, $A \otimes B$ denotes their Kronecker product.
The following result will be used in the paper:
\begin{lemma}[Vectorization] \label{lemma:vectorization}
For matrices $A \in \mathbb{R}^{k \times l}$, $B \in \mathbb{R}^{l \times m}$, $C \in \mathbb{R}^{m \times n}$ and $D \in \mathbb{R}^{k \times n}$,  the matrix equation $ABC = D$ is equivalent to (see,~\cite[Ex. $10.18$ Roth's identity]{abadir05}),  
\begin{subequations}
    \begin{align}
&(C^{\top} \otimes A) \vv{B} = \vv{ABC} = \vv{D},  \label{eq:vectorize1}\\
&\vv{ABC} = (C^{\top}B^{\top} \otimes I_{k}) \vv{A} \label{eq:vectorize2}
\end{align}
\end{subequations}
\end{lemma}
For a better readability, all unknown matrix variables to be computed will be written in a boldface font, \emph{e.g.}, $\W, \K$ etc.  

\section{Problem Setting}\label{sec:prob}

\subsection{Data-generating system and constraints}
We consider the following discrete-time data-generating system 
\begin{equation}\label{eq:system}
     x(k+1)=Ax(k)\!+\!Bu(k)\!+\!w(k),
\end{equation}
where  $x(k) \in \mathbb{R}^{n}$, $u(k) \in \mathbb{R}^{m}$ and $w(k) \in \mathbb{R}^{n}$ are the state, control input and the (additive) disturbance vectors at time $k$, respectively. The system matrices $A \in  \mathbb{R}^{n\times n}$ and $B \in \mathbb{R}^{n\times m}$ are \emph{unknown}. 
A state-input trajectory of $T+1$ samples $\{x(k), u(k)\}_{k=1}^{T+1}$ is generated from system \eqref{eq:system}. The generated data is represented with the following matrices,
\begin{subequations}\label{eq:data}
\begin{align}
    X^{+} &\triangleq [x(2) \quad  x(3) \quad \cdots \quad x(T+1)] \in \mathbb{R}^{n\times T}, \\
    X &\triangleq [x(1) \quad  x(2) \quad \cdots \quad x(T)] \in \mathbb{R}^{n\times T}, \\
    U &\triangleq [u(1) \quad  u(2) \quad \cdots \quad u(T)] \in \mathbb{R}^{m\times T}. 
\end{align}
\end{subequations}

The system \eqref{eq:system} is subject to the following  state, input and disturbance constraints, respectively:
\begin{subequations}\label{eq:constr}
    \begin{align}
\mathcal{X} \triangleq&\left \{  x: Hx\leq \mathbf{1}_{n_x} \right \},\;\;\;\;\;\;\\ 
\mathcal{U} \triangleq&\left \{  u: Gu\leq \mathbf{1}_{n_u} \right \},\;\;\;\;\;\;\;\\ 
\mathcal{W} \triangleq&\left \{  w: -\mathbf{1}_{n_w} \leq Dw\leq \mathbf{1}_{n_w}\right \},\;\;\;
\end{align}
\end{subequations}
where $H \in \mathbb{R}^{n_{x}\times n}$, $G \in \mathbb{R}^{n_{u}\times m}$ and $D \in \mathbb{R}^{n_{w}\times n}$ are given matrices. Note that, the generated state samples in \eqref{eq:data} are \emph{noisy} which are affected by bounded but \emph{unknown} disturbance $w(k) \in \mathcal{W}$ for $k = 1, \ldots, T+1$.

\subsection{Feasible model set and `informative' data}
 We characterize a set of \emph{feasible models} which are compatible with the measured data $X^{+}, X, U$  and the bound on the disturbance samples captured
by the set $\mathcal{W}$, defined as follows
\begin{equation}\label{eq:feasible model set}
\begin{split}
   \mathcal{M} \triangleq &\left \{M  \in \mathbb{R}^{n\times (n+m)}:   \right.\\
&\left.x(k+1) - M \smallmat{x(k) \\ u(k)} \in \mathcal{W},  k=1, \ldots, T   \right \}. 
\end{split}
\end{equation}

Using the definitions of data matrices in \eqref{eq:data} and disturbance set $\mathcal{W}$ in \eqref{eq:constr}, the feasible model set $\mathcal{M} $ is represented as,
\begin{equation}\label{eq:feasible model set 1}
   \mathcal{M} \triangleq \left \{M \in \mathbb{R}^{n\times (n+m)}:   
 -\bar{\mathbf{1}} \leq D X^{+} - D M \smallmat{X \\ U} \leq \bar{\mathbf{1}}  \right \},
\end{equation}
with $\bar{\mathbf{1}} \triangleq \smallmat{\mathbf{1}_{n_w} \ \mathbf{1}_{n_w} \cdots \ \mathbf{1}_{n_w}} \in \mathbb{R}^{n_w\times T}$.

\begin{proposition}[Bounded feasible model set~\cite{bisoffi23,zhong22}]\label{prop:rich_data}
The feasible model set $\mathcal{M}$ in \eqref{eq:feasible model set 1} is a bounded polyhedron if and only if $\mathrm{rank}\left(\smallmat{X \\U} \right) = n+m$ and $D$ has a full column rank. 
\end{proposition}
The above proposition relates to the ``informative" data and \emph{persistency of excitation} conditions~\cite{tesi20}. The full row rank of $\smallmat{X \\U}$ can be checked from the data, if this condition is not satisfied, the set $\mathcal{M}$ is unbounded which makes it difficult to find a feasible controller and RCI set.  

\subsection{RCI set definition and invariance inducing controller}
Let us consider a static state feedback control law
\begin{equation}\label{eq:State Feedback}
u(k)=\K x(k),
\end{equation}
where $\K \in\mathbb{R}^{m\times n}$ is a feedback gain matrix. The resulting closed-loop dynamics  for a feasible model $M \in \mathcal{M}$ (using \eqref{eq:feasible model set} and \eqref{eq:State Feedback}) is
\begin{equation}\label{eq:Controlled System}
x^+ = M \smallmat{I \\ \K}x+w,
\end{equation}
where the $k$ dependence is dropped and $x(k+1)$ is denoted as $x^+$ for convenience. 

Let us consider the following polytopic set\footnote{In the definition of $\mathcal{C}$, we have assumed that $\W$ is invertible, which would be later guaranteed by the LMI conditions for invariance. The matrix $P$ is a-priori chosen by the user.  The selection of $P$ is discussed in details in~\cite{ag20}.}  \vspace{-0.2cm}
\begin{equation}\label{eq:invset1}
\mathcal{C} \triangleq\left \{ x\in\mathbb{R}^{n} : -\mathbf{1}_{n_p}\leq P\W^{-1}x \leq \mathbf{1}_{n_p} \right \},
\end{equation}
where $P \in \mathbb{R}^{n_p\times n}$, $\W \in \mathbb{R}^{n\times n}$.

The set~$\mathcal{C}$ is referred to as \emph{robustly invariant} for the system~\eqref{eq:Controlled System}, if the following condition is satisfied:\vspace{-0.2cm}
\begin{equation}
\label{eq:Invariance Condition}
x\in\mathcal{C}  \;\Rightarrow\; x^+\in\mathcal{C},\; \forall w\in\mathcal{W}, \; \; \forall M \in \mathcal{M}.
\end{equation}

\begin{remark}
Note that, the invariance condition \eqref{eq:Invariance Condition} has to be satisfied \emph{for all} feasible models $M \in \mathcal{M}$, compatible with data. This can be seen  as a data-based counterpart of the  \emph{model-based} invariance condition given by
$x\in\mathcal{C}  \;\Rightarrow\; (A+B \K)x + w \in\mathcal{C},\; \forall w\in\mathcal{W},$ wherein exact knowledge of the system matrices $A, B$ is assumed. In other words,  the data-based invariance condition \eqref{eq:Invariance Condition} aims at designing $K$ robustly for the set $\mathcal{M}$ induced due to the lack of model knowledge caused by the disturbances $w(k)$ and \emph{finite} data samples. This condition is also the main difference w.r.t. the approach presented in \cite{sam22}, wherein a single feasible model is sought while using a model-dependent disturbance set to take into account finiteness of the data.  
\end{remark}

The set $\mathcal{C}$ has to satisfy the state and input constraints, this implies $\mathcal{C}\subseteq\mathcal{X}$ and $K\mathcal{C}\subseteq\mathcal{U}$, which can be further expressed as\vspace{-0.2cm}
\begin{align}\label{eq:State Constraint}
x\in\mathcal{C} \;&\Rightarrow\; x\in\mathcal{X},\\\label{eq:Control Input Constraint} 
x\in\mathcal{C} \;&\Rightarrow\; u=\mathbf{K}x\in\mathcal{U}.
\end{align}
 
The problem considered in this paper is formalized as follows:

\begin{problem}\label{prob:Problem Formulation}
Given data matrices $(X^+, X,U)$ defined in (\ref{eq:data}), the  constraints sets (\ref{eq:constr}) and a fixed matrix $P$, find the matrix $\W$ defining the invariant set  $\mathcal{C}$  in \eqref{eq:invset1} and a feedback controller gain $\K$ such that: 
\begin{enumerate}
\item The invariance condition (\ref{eq:Invariance Condition}) holds;
\item All elements of the set $\mathcal{C}$  satisfy the state and input constraints (\ref{eq:State Constraint}) and (\ref{eq:Control Input Constraint}), respectively.
\end{enumerate}
\end{problem}\vspace{0cm}

We aim at maximizing the volume of set $\mathcal{C}$ solving Problem~\ref{prob:Problem Formulation}. 



\section{Tractable formulations for System Constraints and Invariance Condition}\label{sec:LMI}

In this section, we present a convenient coordinate transformation~\cite{gkf17} such that state and control input constraints (\ref{eq:State Constraint})-(\ref{eq:Control Input Constraint}) are expressed as affine inequalities, while the invariance condition  (\ref{eq:Invariance Condition})  is expressed as a set of LMIs. Thus, a solution to Problem~\ref{prob:Problem Formulation} is given in the form of an LMI feasibility problem. The volume maximization of the invariant set is then performed via a semi-definite programming problem.

\subsection{System constraints}
Let us consider the following state transformation 
\vspace{-0.1cm}\begin{equation}\label{eq:trans}
\theta = W^{-1}x \Leftrightarrow x=W\theta.\vspace{-0.1cm}
\end{equation}
 This allows us to express the set $\mathcal{C}$ in \eqref{eq:invset1} as
\vspace{-0.1cm}\begin{equation}
\label{eq:ninvset}
\mathcal{C} \triangleq \left \{ \W \theta \in \mathbb{R}^n:\theta \in \Theta \right \},\vspace{-0.1cm}
\end{equation}
where $\Theta$ is a symmetric set defined as follows:
\vspace{-0.1cm}\begin{equation}
\label{eq:Theta}
\Theta\triangleq\left \{ \theta\in\mathbb{R}^{n}: -\mathbf{1}_{n_p} \leq P\theta \leq \mathbf{1}_{n_p} \right \}.\vspace{-0.1cm}
\end{equation}
Note that in the $\theta$-state-space, the candidate invariant set $\Theta$ is a \emph{known} symmetric set  around the origin. The corresponding polytopic set $\mathcal{C}$ in the $x$-state-space will be completely determined by the choice of $\W$, which we aim to compute.

As $P$ is a known matrix, the symmetric set $\Theta$ can be  expressed as the convex hull of the finitely many \emph{known} vertices $\left\{\theta^1,\ldots,\theta^{2\sigma }\right\}$:
\begin{equation}
\label{eq:Theta - Convex Hull}
\Theta=\mathsf{conv}\left(\left\{\theta^1,\ldots,\theta^{2\sigma }\right\}\right),
\end{equation}
where $\sigma$ is some known positive integer determined by the choice of $P$. 
  We now express the state and input inequality constraints  \eqref{eq:constr} in the $\theta$-state-space by using the transformation \eqref{eq:trans}. 
  Satisfaction of these inequalities constraints at the vertices $\{\theta^j\}_{j=1}^{2\sigma}$ ensures that they are satisfied over the whole set $\Theta$ as well. Therefore, we can write the state constraints \eqref{eq:State Constraint} in terms of $\W$ as follows:
\vspace{-0.1cm}\begin{equation}\label{eq:Tractable State Constraints}
H\W\theta\leq \mathbf{1}_{n_x},\forall\theta\in\Theta \;\Leftrightarrow\; H\W\theta^j\leq \mathbf{1}_{n_x},\;j\!=\!1,\ldots,2\sigma.\vspace{-0.1cm}
\end{equation}
In order to express the control input constraints in terms of $\W$, let us consider a new matrix variable as follows:
\begin{equation}
\label{eq:N}
\N \triangleq \K\W \;\;\Leftrightarrow\;\; \K = \N\W^{-1}.\vspace{-0.1cm}
\end{equation}
The control input constraints in \eqref{eq:Control Input Constraint} are then given as 
\begin{equation}\label{eq:Tractable Control Input Constraints}
G \N \theta\leq \mathbf{1}_{n_u},\forall\theta\in\Theta \;\Leftrightarrow\; G \N \theta^j\leq \mathbf{1}_{n_u},\;j=1,\ldots,2\sigma.
\end{equation}
The system constraints (\ref{eq:Tractable State Constraints}) and (\ref{eq:Tractable Control Input Constraints}) are affine and are identified by $n_x\times2\sigma$ and $n_u\times2\sigma$ scalar inequalities. These constraints are exact, in contrast to the relaxation used in \cite{lc15}. Hence, no conservatism is introduced. 
Note that, if the sets~$\mathcal{X}$ and~$\mathcal{U}$ are symmetric, half of the constraints in~(\ref{eq:Tractable State Constraints}) and~(\ref{eq:Tractable Control Input Constraints}) are redundant and hence removable. This is a consequence of the symmetry of the set $\Theta$, which allows arranging~$\theta^j$'s in a way to have~$\theta^{j+\sigma}=-\theta^j$ for $j=1,\ldots,\sigma$.

\subsection{Invariance conditions in the transformed state-space}
Before we state the condition for invariance, let us express the system dynamics  in the $\theta$-state-space. Using \eqref{eq:trans},  the closed-loop dynamics \eqref{eq:Controlled System} can be written as 
\begin{equation}\label{eq:closeloop}
\W \theta^{+}=M \begin{bmatrix}
\W\\ \N 
\end{bmatrix} \theta + w,
\end{equation}
for a feasible model $M \in \mathcal{M}$ and $w \in \mathcal{W}$. 

We now state two equivalent invariance conditions in the $\theta$ state-space based on the closed-loop dynamics \eqref{eq:closeloop}.  

\begin{lemma}
If the set $\Theta$ in \eqref{eq:Theta} is robustly invariant for system \eqref{eq:closeloop} then the following two statements are equivalent:
\begin{enumerate}
    \item[(i)]  for all $\theta \in \Theta$, $\forall (w, M) \in (\mathcal{W}, \mathcal{M})$,
    \begin{equation}\label{eq:inv_cond1}
        \theta^{+} = \left( \W^{-1} M \begin{bmatrix}
\W\\ \N 
\end{bmatrix} \theta + \W^{-1}w \right) \in \Theta
    \end{equation}
\item[(ii)] for each vertex $\theta^{j}, j = 1, \ldots, 2\sigma$ of the set $\Theta$,  and $\forall(w, M) \in (\mathcal{W}, \mathcal{M})$,  
\begin{equation}\label{eq:inv_cond2}
\tjplus = \left( \W^{-1} M \begin{bmatrix}
\W\\ \N 
\end{bmatrix} \theta^{j} + \W^{-1}w \right) \in \Theta
\end{equation}
\end{enumerate}
\end{lemma}
\vspace{0.1cm}
\begin{proof}
 Since for each vertex $\theta^{j}$, it holds that $\theta^{j} \in \Theta$,  it can be easily seen that $(i) \Rightarrow (ii)$. Let us now prove the converse statement, \emph{i.e.},  $(ii) \Rightarrow (i)$.  
From \eqref{eq:Theta - Convex Hull}, any $\theta \in \Theta$ can be expressed as a convex combination of the vertices, 
 \begin{equation}
     \theta = \sum_{j=1}^{2\sigma} \alpha_j \theta^{j}, \quad \sum_{j=1}^{2\sigma} \alpha_j = 1, \quad \alpha_j \in [0,1].
 \end{equation}
Then, based on the closed-loop dynamics \eqref{eq:closeloop} we get,
\begin{align}\label{eq:cmb}
    \theta^{+} &= \W^{-1}M \begin{bmatrix}
\W\\ \N 
\end{bmatrix}\left(\sum_{j=1}^{2\sigma} \alpha_j \theta^{j}\right) + \W^{-1}w \nonumber \\
&= \left(\sum_{j=1}^{2\sigma} \alpha_j \W^{-1}M \begin{bmatrix}
\W\\ \N 
\end{bmatrix}\theta^{j}\right) + \underbrace{\left(\sum_{j=1}^{2\sigma}\alpha_j\right)}_{1} \W^{-1}w \nonumber \\
&= \sum_{j=1}^{2\sigma} \alpha_j \underbrace{ \left(\W^{-1} M \begin{bmatrix}
\W\\ \N 
\end{bmatrix} \theta^{j} + \W^{-1}w \right)}_{(\theta^j)^{+}}
\end{align}
We know that $\tjplus \in \Theta, \forall(w, M) \in (\mathcal{W}, \mathcal{M})$ according to \eqref{eq:inv_cond2}. Since $\theta^{+}$ in \eqref{eq:cmb} is obtained as a convex combination of $\tjplus$ and as the set $\Theta$ is convex, it necessarily follows that $\theta^{+} \in \Theta \ \forall(w, M) \in (\mathcal{W}, \mathcal{M})$, thus proving $(ii) \Rightarrow (i)$.
\end{proof}
In the rest of the paper, we will consider  condition \eqref{eq:inv_cond2} for robust invariance of the set $\Theta$. Note that, eq.~\eqref{eq:inv_cond2} allows us to enforce the invariance condition only at a finite set of known vertices, instead of enforcing it for all $\theta \in \Theta$.

\subsection{Data-based LMI condition for invariance}

We will now state and prove a data-based sufficient condition to render the set $\Theta$ invariant with an associated state-feedback controller. Recall that  a $T+1$-length state-input trajectory $\{x(k), u(k)\}_{k=1}^{T+1}$  generated from system \eqref{eq:system} is available. The data is arranged in the form of matrices $(X^+, X,U)$ as in \eqref{eq:data}. Let us first define the following matrix and a vector, which are constructed from the given state-input data and a known disturbance set matrix $D$ in \eqref{eq:constr}.
\begin{subequations}\label{eq:Zd}
    \begin{align}
Z &\triangleq \left( \begin{bmatrix}
    X \\ U
\end{bmatrix}^{\top} \otimes D \right)  \in \mathbb{R}^{Tn_{w} \times n(n+m)}\\ 
  d &\triangleq \begin{bmatrix}
      Dx(2) \\ Dx(3) \\ \vdots \\ Dx(T+1)
  \end{bmatrix} \in \mathbb{R}^{Tn_w}
\end{align}
\end{subequations}

The following theorem states the data-based sufficient LMI feasibility condition for invariance and control.  

\begin{theorem}[Data-based LMI for invariance]\label{thm:data-based invariance}
Given data matrices $(X^+, X,U)$ and a fixed matrix $P \in \mathbb{R}^{n_p \times n}$, if there exists  $\W \in \mathbb{R}^{n \times n}$, $\N \in \mathbb{R}^{m \times n}$, and the variables 
$\{\bm{\phi}_{ij} \in \mathbb{R}_{+},  \bm{\Lambda}_{ij} \in \mathbb{D}^{Tn_w}_{+}, \bm{\Gamma}_{ij} \in \mathbb{D}^{n_w}_{+}\}$ that satisfy the following LMIs for $i=1,\ldots, n_p$ and $j=1,\ldots, 2\sigma$,
\begin{equation}\label{eq:DB invariance condition}
    \begin{bmatrix}\!
  \bm{r}_{ij} & -d^{\!\top}\bm{\Lambda}_{ij}Z & \bm{0}  & \bm{0}\\ 
* &  Z^{\top}\bm{\Lambda}_{ij}Z   &\bm{0} & \mathcal{G}^{\top}\!\!\left(\W, \N, \theta^{j}\right)\\ 
* & * & D^{\top} \bm{\Gamma}_{ij} D \!& I_n\\
*&*&*&\W\!\!+\!\!\W^{\top} \!\!-\!\! \bm{\phi}_{i,j} P^{\!\top}e_ie^{\!\top}_i P
\!\end{bmatrix}{\succcurlyeq} 0,
\end{equation}
where,
\begin{subequations}
 \begin{align}
    &\bm{r}_{ij} \triangleq \bm{\phi}_{ij}\!-\!\mathbf{1}^{\top}\!\bm{\Lambda}_{ij}\mathbf{1}-\!\mathbf{1}_{n_w}^{\top}\!\bm{\Gamma}_{ij}\mathbf{1}_{n_w} +d^{\top}\bm{\Lambda}_{ij}d \ \ \in \mathbb{R}, \label{eq:rij} \\
     &\mathcal{G}\left(\W, \N, \theta^{j}\right) \triangleq  \left(\begin{bmatrix}
        \W \\ \N
    \end{bmatrix} \theta^{j}\right)^{\top} \!\!\otimes\! I_n  \ \ \in \mathbb{R}^{n \times n(n+m)}, \label{eq:G}
\end{align}   
\end{subequations}
then, the state feedback controller gain is obtained as $\K = \N \W^{-1}$ which renders the set $\mathcal{C}$ in \eqref{eq:ninvset} robust invariant.  \end{theorem}
\begin{proof}
We first rewrite the feasible model set $\mathcal{M}$ in \eqref{eq:feasible model set 1} using the vectorization Lemma~\ref{lemma:vectorization} as follows,
\begin{align}\label{eq:Feasible model set vector}
    \mathcal{M} &\triangleq \left\{\vv{M} \in \mathbb{R}^{n(n+m)}:  \right. \nonumber \\
  & \left. \smallmat{\!-\mathbf{1}_{n_w} \!+\! Dx(2) \\ \vdots \\ -\mathbf{1}_{n_w}\!+\!Dx(T\!+\!1) \!}
   \! \! \leq \!\!
    \left( \begin{bmatrix}
    X \\ U 
\end{bmatrix}^{\!\!\top} \!\! \otimes \! D \right) \vv{M}\! \leq \! \smallmat{\!\mathbf{1}_{n_w} \!+\! Dx(2)  \\ \vdots \\ \mathbf{1}_{n_w}+Dx(T\!+\!1) \!}
      \right\}  \nonumber  \\
    &\triangleq \left\{\vv{M} \in \mathbb{R}^{n(n+m)}: -\mathbf{1}_{Tn_w} \!+\!d  \leq  Z \vv{M} \leq \!\mathbf{1}_{Tn_{w}}\!+\!d \right\},
\end{align}
where we have used the identity~\eqref{eq:vectorize1} to rewrite the inequalities  in \eqref{eq:feasible model set 1} in a vector form and substituted  $Z, d$ as defined in \eqref{eq:Zd}. Recall that $\vv{M}$ denotes the vectorization of model matrix $M$ obtained by stacking its column vectors. 

Similarly, using the identity \eqref{eq:vectorize2}, we rewrite the  closed-loop dynamics \eqref{eq:closeloop} at the vertex $\theta^{j}$  as follows,
\begin{align}\label{eq:G_dynamics}
\W \tjplus= \underbrace{ \left(\left(\begin{bmatrix}
        \W \\ \N
    \end{bmatrix} \theta^{j}\right)^{\top}\otimes I_n  \right)}_{\mathcal{G}\left(\W, \N, \theta^{j}\right)}\vv{M} + w.   
\end{align}

From \eqref{eq:Theta} the invariance condition in \eqref{eq:inv_cond2} can be written as, 
for all $i=1,\ldots,n_p, \; j=1,\ldots, 2\sigma$, 
\begin{equation}\label{LMI Starting Point}
1 - (e_i^{\top}P\tjplus)^{2}\geq 0,\; \forall w \in \mathcal{W}, \; \forall \vv{M} \in \mathcal{M},
\end{equation} 
where  $e_i$ is the $i$-th column vector of the identity matrix $I_{n_p}$. 

We now multiply \eqref{LMI Starting Point} by a positive scalar variable $\bm{\phi}_{ij} >0$ and lower bound the left hand side by a term that is known to be non-negative for all $w \in \mathcal{W}, \; \vv{M} \in \mathcal{M}$ (S-procedure~\cite{cs97}). 
In this way, we obtain a sufficient condition for invariance as follows,
\begin{multline} \label{eq:relax_ineq}
\bm{\phi}_{ij}(1 - (e_i^{\top}P\tjplus)^2) \geq \\
2\left(\tjplus\right)^{\top}\!\underbrace{\left(\mathcal{G}(\W, \N, \theta^{j})\vv{M}+w-\W \tjplus \right)}_{0} \\ + 
\underbrace{\left( (\mathbf{1}+d)-Z\vv{M} \right)^{\top}\bm{\Lambda}_{ij}\left((\mathbf{1}-d)+Z\vv{M}\right)}_{\geq0} \\
+\underbrace{(\mathbf{1}+Dw)^{\top}\bm{\Gamma}_{ij}(\mathbf{1}-Dw)}_{\geq0},\vspace{-0.2cm}
\end{multline}
with 
$\bm{\Lambda}_{ij} \in \mathbb{D}^{Tn_w}_{+},  \bm{\Gamma}_{ij} \in \mathbb{D}^{n_w}_{+}$, being diagonal matrices having non-negative entries.
Based on \eqref{eq:G_dynamics} and the set definitions  $\mathcal{W}$, $\mathcal{M}$ in \eqref{eq:constr}, \eqref{eq:Feasible model set vector} respectively, it is straightforward to verify that the right hand side of \eqref{eq:relax_ineq} is nonnegative. 

A sufficient invariance condition is obtained by re-arranging \eqref{eq:relax_ineq} into the following quadratic form,
\begin{equation}\label{eq:quad_form}
\varkappa^{\top}\mathcal{P}_{ij}(\W,\N,\bm{\Lambda}_{ij},\bm{\Gamma}_{ij},\bm{\phi}_{ij})\varkappa\succcurlyeq 0,\;\forall \varkappa, 
\end{equation}
where $\varkappa^{\top}=\begin{bmatrix}
1 &\vv{M}^{\top} & w^{\top} &-(\tjplus)^{\top}
\end{bmatrix}$ and $\mathcal{P}_{ij}$ is a symmetric matrix given by,
\begin{equation}\label{eq:MI invariance condition}
    \begin{bmatrix}
  \bm{r}_{ij} & -d^{\top}\bm{\Lambda}_{ij}Z & \bm{0}  & \bm{0}\\ 
* &  Z^{\top}\bm{\Lambda}_{ij}Z   &\bm{0} & \mathcal{G}^{\top}\left(\W, \N, \theta^{j}\right)\\ 
* & * & D^{\top} \bm{\Gamma}_{ij} D \!& I_n\\
*&*&*&\W\!+\!\W^{\top} \!-\! \bm{\phi}_{i,j} P^{\top}e_ie^{\top}_i P
\end{bmatrix},
\end{equation}
where $\bm{r}_{ij}  \in \mathbb{R}$ is as given in \eqref{eq:rij}
and  $*$'s represent entries that are uniquely identifiable from symmetry. 
The invariance condition~\eqref{eq:inv_cond2}  holds if $\mathcal{P}_{ij}\succcurlyeq0$. The statement of  Theorem~\ref{thm:data-based invariance} thus follows. 
\end{proof}

\subsection{Dilated data-based LMI condition for invariance}

In this subsection, we derive a set of modified data-based LMI
conditions for invariance. These LMIs have additional matrix variables and are potentially less conservative than those introduced in Theorem~\ref{thm:data-based invariance}.

Let us introduce new matrix variables $\V_{ij} \in \mathbb{R}^{n \times n}$ and signals $\xi_{ij} = \V_{ij}^{-1} \W \tjplus$, for $i=1,\ldots, n_p$ and $j=1, \ldots, 2\sigma$. From the dynamics~\eqref{eq:G_dynamics} we obtain,
\begin{equation}
    \mathcal{G}\left(\W, \N, \theta^{j} \right)\vv{M} +w - \V_{ij} \xi_{ij} = 0.
\end{equation}
The sufficient condition in \eqref{eq:relax_ineq} is now expressed in the new introduced variables as follows:
\begin{multline} \label{eq:relax_ineq_dilated}
\bm{\phi}_{ij}(1 - (e_i^{\top}P\W^{-1} \V_{ij}\xi_{ij})^2) \geq \\
2\xi_{ij}^{\top}\!\underbrace{\left(\mathcal{G}(\W, \N, \theta^{j})\vv{M}+w- \V_{ij} \xi_{ij} \right)}_{0} \\ + 
\underbrace{\left( (\mathbf{1}+d)-Z\vv{M} \right)^{\top}\bm{\Lambda}_{ij}\left((\mathbf{1}-d)+Z\vv{M}\right)}_{\geq0} \\
+\underbrace{(\mathbf{1}+Dw)^{\top}\bm{\Gamma}_{ij}(\mathbf{1}-Dw)}_{\geq0}.\vspace{-0.2cm}
\end{multline}
Along the similar lines as described in the previous subsection,  a sufficient condition for invariance is obtained by re-arranging \eqref{eq:relax_ineq_dilated} into the following quadratic form:
\begin{equation}\label{eq:quad_form}
\varkappa^{\top}\mathcal{P}_{ij}(\W,\N,\bm{\Lambda}_{ij},\bm{\Gamma}_{ij},\bm{\phi}_{ij}, \V_{ij})\varkappa\succcurlyeq 0,\;\forall \varkappa, 
\end{equation}
where $\varkappa^{\top}=\begin{bmatrix}
1 &\vv{M}^{\top} & w^{\top} &-\xi_{ij}^{\top}
\end{bmatrix}$ and $\mathcal{P}_{ij}$ is a symmetric matrix. The invariance condition thus holds if $\mathcal{P}_{ij}\succcurlyeq0$, \emph{i.e.},
\begin{equation}\label{eq:LMI invariance condition dilated}
    \begin{bmatrix}
  \bm{r}_{ij} & -d^{\top}\bm{\Lambda}_{ij}Z & \bm{0}  & \bm{0}\\ 
* &  Z^{\top}\bm{\Lambda}_{ij}Z   &\bm{0} & \mathcal{G}^{\top}\left(\W, \N, \theta^{j}\right)\\ 
* & * & D^{\top} \bm{\Gamma}_{ij} D \!& I_n\\
*&*&*&\V_{ij}\!+\!\V_{ij}^{\top}\!-\!\V_{ij}^{\top}\mathcal{L}_{ij} \V_{ij}
\end{bmatrix} \succcurlyeq 0
\end{equation} 
where $\mathcal{L}_i \triangleq \bm{\phi}_{ij} \W^{-\top}P^{\top}e_{i}e_{i}^{\top}P \W^{-1}$ and  $\bm{r}_{ij}$, $\mathcal{G}(\W, \N, \theta^{j})$ are as defined in \eqref{eq:rij}, \eqref{eq:G} respectively. Note that the block $(4,4)$ in \eqref{eq:LMI invariance condition dilated} has a nonlinear dependence on $ \bm{\phi}_{i,j}, \V_{i,j} $ and $\W$, which will be resolved by introducing new matrix variables. 
We will now state the following dilated sufficient LMI conditions for invariance. 

\begin{theorem}[Dilated LMI conditions for invariance]\label{thm:data-based invariance dilated}
    Given data matrices $(X^+, X,U)$ and a fixed matrix $P \in \mathbb{R}^{n_p \times n}$, if there exists  $\W \in \mathbb{R}^{n \times n}$, $\N \in \mathbb{R}^{m \times n}$, and  variables 
$\{ \bm{\phi}_{ij} \in \mathbb{R}_{+},  \bm{\Lambda}_{ij} \in \mathbb{D}^{Tn_w}_{+}, \bm{\Gamma}_{ij} \in \mathbb{D}^{n_w}_{+}, \bm{X}_{ij}, \V_{i,j} \in \mathbb{R}^{n \times n}\}$ that satisfy the following LMIs for $i=1,\ldots, n_p$ and $j=1,\ldots, 2\sigma$,
\begin{equation}\label{eq:Dilated LMI 1}
    \begin{bmatrix}
         \W^{\top} + \W - \bm{X}_{ij} & \bm{\phi}_{ij} P^{\top}e_{i} \\
         \bm{\phi}_{ij} e_{i}^{\!\top}P   & \bm{\phi}_{ij}
    \end{bmatrix} {\succcurlyeq} 0.
\end{equation}
\begin{equation}\label{eq:Dilated LMI 2}
\begin{bmatrix}
  \bm{r}_{ij} & -d^{\top}\bm{\Lambda}_{ij}Z & \bm{0}  & \bm{0} & \bm{0}\\ 
* &  Z^{\top}\bm{\Lambda}_{ij}Z   &\bm{0} & \mathcal{G}^{\top}\left(\W, \N, \theta^{j}\right) & \bm{0} \\ 
* & * & D^{\top} \bm{\Gamma}_{ij} D \!& I_n & \bm{0}\\
*&*&*&\V_{ij}\!+\!\V_{ij}^{\top} & \V^{\top}_{ij} \\
*&*&*& * & \bm{X}_{ij}
\end{bmatrix} \succcurlyeq 0,
\end{equation}
where, $\bm{r}_{ij}$, $\mathcal{G}(\W, \N, \theta^{j})$ are as defined in \eqref{eq:rij}, \eqref{eq:G},
then, the state feedback controller gain is obtained as $\K = \N \W^{-1}$ which renders the set $\mathcal{C}$ in \eqref{eq:ninvset} robust invariant.  \end{theorem}
\begin{proof}
In order to resolve the non-linearity in the block $(4,4)$ of \eqref{eq:LMI invariance condition dilated}, let us introduce a new matrix variable $\bm{X}_{ij}= \bm{X}^{\top}_{ij} \succ 0$ such that
\begin{equation}\label{eq:Xi Original}
  \bm{X}_{ij}^{-1}\!-\!\mathcal{L}_{ij} {\succ} 0 \Leftrightarrow  \bm{X}_{ij}^{-1}\!-\!\bm{\phi}_{ij} \W^{-\top}P^{\top}e_{i}e_{i}^{\!\top}P \W^{-1}{\succ} 0. 
\end{equation} 
By applying Schur complement to \eqref{eq:Xi Original} we have,
\begin{equation}\label{eq:Xi Schur}
    \begin{bmatrix}
         \bm{X}^{-1}_{ij} & \bm{\phi}_{ij} \W^{-\top}P^{\top}e_{i} \\
         \bm{\phi}_{ij} e_{i}^{\!\top}P \W^{-1}  & \bm{\phi}_{ij}
    \end{bmatrix} {\succ} 0.
\end{equation}
Using congruence transformation matrix $\mathrm{diag}\{\W, I_n\}$, \eqref{eq:Xi Schur} can be rewritten as
\begin{equation}\label{eq:Xi lmi}
    \begin{bmatrix}
         \W^{\top}\bm{X}^{-1}_{ij}\W & \bm{\phi}_{ij} P^{\top}e_{i} \\
         \bm{\phi}_{ij} e_{i}^{\!\top}P   & \bm{\phi}_{ij}
    \end{bmatrix} {\succ} 0.
\end{equation}
In order to resolve the nonlinear dependence in the $(1,1)$ block of the left hand side matrix in \eqref{eq:Xi lmi}, we use,
\begin{multline}\label{eq:linearization W}
 \W^{\top}\bm{X}^{-1}_{ij}\W \!=\! (\W \! - \!  \bm{X}_{ij})^{\! \top} \bm{X}^{-1}_{ij}  (\W \! -\!   \bm{X}_{ij}) \! + \! \W + \W^{\!\top} \! - \!  \bm{X}_{ij} \\
  \succcurlyeq \W + \W^{\top} -  \bm{X}_{ij}
\end{multline}
From this inequality, replacing $\W^{\top}\bm{X}^{-1}_{ij}\W$ in \eqref{eq:Xi lmi} with $\W + \W^{\top} -  \bm{X}_{ij}$, leads to a sufficient LMI condition for \eqref{eq:Xi lmi} as in \eqref{eq:Dilated LMI 1}. Thus, proving the first LMI condition~\eqref{eq:Dilated LMI 1} stated in Theorem~\ref{thm:data-based invariance dilated}.

From \eqref{eq:Xi Original}, the condition \eqref{eq:LMI invariance condition dilated} can be rewritten as
\begin{equation}\label{eq:invariance condition dilated Xi}
    \begin{bmatrix}
  \bm{r}_{ij} & -\!d^{\top}\bm{\Lambda}_{ij}Z & \bm{0}  & \bm{0}\\ 
* &  Z^{\top}\bm{\Lambda}_{ij}Z   &\bm{0} & \mathcal{G}^{\!\top}\left(\!\W, \N, \theta^{j}\!\right)\\ 
* & * & D^{\top} \bm{\Gamma}_{ij} D \!& I_n\\
*&*&*&\V_{ij}\!+\!\V_{ij}^{\!\top}\!\!-\!\!\V_{ij}^{\!\top}\bm{X}^{-1}_{ij} \V_{ij}
\end{bmatrix} \!\!\succcurlyeq \!\!0,
\end{equation} 
which followed by Schur complement can be written as
\begin{equation}
    \begin{bmatrix}
  \bm{r}_{ij} & -d^{\top}\bm{\Lambda}_{ij}Z & \bm{0}  & \bm{0} & \bm{0}\\ 
* &  Z^{\top}\bm{\Lambda}_{ij}Z   &\bm{0} & \mathcal{G}^{\top}\left(\W, \N, \theta^{j}\right) & \bm{0} \\ 
* & * & D^{\top} \bm{\Gamma}_{ij} D \!& I_n & \bm{0}\\
*&*&*&\V_{ij}\!+\!\V_{ij}^{\top} & \V^{\top}_{ij} \\
*&*&*& * & \bm{X}_{ij}
\end{bmatrix} \! \succcurlyeq \! 0,
\end{equation} 
proving the second LMI condition \eqref{eq:Dilated LMI 2}. 
\end{proof}

\section{Computation of RCI set with volume maximization}\label{sec:SDP}

In this section, we present our algorithm  to compute RCI sets of desirably large size. To this end, we combine state, input constraints and LMI invariance conditions derived in the previous section in a semi-definite programming (SDP) problem to maximize the volume of the RCI set.  

\subsection{One-step algorithm:}
We note that the volume of the invariant set $\mathcal{C}$ in \eqref{eq:invset1} is proportional to the determinant $|\mathrm{det}(\W)|$~\cite{ag20}. Moreover, the RCI set is required to satisfy the state constraints \eqref{eq:Tractable State Constraints}, control input constraints \eqref{eq:Tractable Control Input Constraints} as well as data-based LMI conditions for invariance \eqref{eq:Dilated LMI 1}-\eqref{eq:Dilated LMI 2} (or \eqref{eq:DB invariance condition}). 
Under these constraints, we can easily formulate a determinant
maximization problem. Thus, Problem~\ref{prob:Problem Formulation} is feasible if the following SDP program has a feasible solution,

\textbf{Algorithm 1:}
\begin{equation}\label{eq:one-step SDP}
\begin{array}{lll}
\max & \mathrm{log}\,\mathrm{det}(\W) & \\
{\bm{Z}_{\mathrm{SDP}}} & \\
\text{subject to:} & \W = \W^{\top}, &\\
& 
 (\ref{eq:Tractable State Constraints}),  (\ref{eq:Tractable Control Input Constraints}), &  (\text{state-input constraints}) \\
 &   \eqref{eq:Dilated LMI 1}\!-\!\eqref{eq:Dilated LMI 2}\; (\text{or} \;  \eqref{eq:DB invariance condition}) & (\text{invariance LMIs})
\end{array}
\end{equation}
where the optimization variables are $\bm{Z}_{\mathrm{SDP}} \triangleq \left( \W,\N,\bm{X}_{ij},\V_{ij},\bm{\phi}_{ij},\bm{\Lambda}_{ij}, \bm{\Gamma}_{ij} \right)$ for $i=1,\ldots, n_p,  j=1,\ldots, 2\sigma$. 
The symmetry condition $\W = \W^{\top}$ is imposed to make the objective function $\mathrm{log}\,\mathrm{det}(\W)$ concave.
The above SDP is a very simple one-step procedure to compute RCI sets of desirably large size. 

\subsection{Iterative algorithm:}

We remark that instead of one-step solution of \eqref{eq:one-step SDP}, an \emph{iterative} volume maximization scheme can be applied to compute the RCI set.  In this approach, the SDP~\eqref{eq:one-step SDP} is solved with an iterative procedure such that the solution obtained at the $q$-th iteration is utilized in the problem to be solved at the $(q+1)$-th iteration in order to reduce conservatism. In such iterative scheme, the $\W$ is not required to be symmetric and the conservatism introduced due to the linearization \eqref{eq:linearization W} can be also be reduced.   

Let $W^q$ and ${X}^q_{ij}$ denote the values of the variables $\W$,  $\bm{X}_{ij}$ obtained at the $q$-th iteration.  
In order to ensure that at each iteration the volume of the RCI set increases, \emph{i.e.}, $|\mathrm{det}(W^{q+1})| \geq |\mathrm{det}(W^{q})|$, we impose the following,
\begin{equation}\label{eq:Wobj}
    \W^{\top} W^{q} + (W^{q})^{\top}\W -(W^{q})^{\top}W^{q} \succcurlyeq \W_{\mathrm{obj}} \succ 0,
\end{equation}
where $\W_{\mathrm{obj}} = \W^{\top}_{\mathrm{obj}} \in \mathbb{R}^{n \times n}$ is the new symmetric matrix variable. 

Moreover, the non-linearity \eqref{eq:linearization W} can be written as,
\begin{align}\label{eq:linearization W new}
 \W^{\top}\bm{X}^{-1}_{ij}\W \! \succcurlyeq \!
   \W^{\top} Z^{q}_{ij} \!+\! (Z^{q}_{ij})^{\top} \W 
   \!-\! (Z^{q}_{ij})^{\top} \bm{X}_{ij} Z^{q}_{ij},
\end{align}
where $Z^{q}_{ij} \triangleq ({X}^q_{ij})^{-1} W^{q}$.

Thus, the $(1,1)$-block in \eqref{eq:Dilated LMI 1} is replaced with the right hand side of \eqref{eq:linearization W new} at the $q$-th iteration as follows,
\begin{equation}\label{eq:Dilated LMI 1 new}
    \begin{bmatrix}
         \W^{\top} Z^{q}_{ij} + (Z^{q}_{ij})^{\top}\W - (Z^{q}_{ij})^{\top} \bm{X}_{ij} Z^{q}_{ij} & \bm{\phi}_{ij} P^{\top}e_{i} \\
         \bm{\phi}_{ij} e_{i}^{\!\top}P   & \bm{\phi}_{ij}
    \end{bmatrix} {\succcurlyeq} 0.
\end{equation}

For brevity, we omit the detailed proof of the iterative algorithm. The reader is referred to~\cite{gkf17} for the details.

The iterative algorithm is summarized as follows:

\textbf{Algorithm 2: $q$-th iteration:}
\begin{equation}\label{eq:iterative SDP}
\begin{array}{lll}
\max & \mathrm{log}\,\mathrm{det}(\W_{\mathrm{obj}}) & \\
{\bm{Z}_{\mathrm{SDP}}} & \\
\text{subject to:} & \eqref{eq:Wobj}, &\\
& 
 (\ref{eq:Tractable State Constraints}),  (\ref{eq:Tractable Control Input Constraints}), &  (\text{state-input constraints}) \\
 &  \eqref{eq:Dilated LMI 2} \ \eqref{eq:Dilated LMI 1 new}, \;  & (\text{invariance LMIs})  
\end{array}
\end{equation}
where the optimization variables are $\bm{Z}_{\mathrm{SDP}} \triangleq \left( \W,\N,\bm{X}_{ij},\V_{ij},\bm{\phi}_{ij},\bm{\Lambda}_{ij}, \bm{\Gamma}_{ij}, \W_{\mathrm{obj}} \right)$ for $i=1,\ldots, n_p,  j=1,\ldots, 2\sigma$.

\section{Numerical Example}\label{sec:example}
We demonstrate the effectiveness of the proposed approach via a numerical case study. All algorithms have been implemented in the \texttt{Python} environment using \texttt{cvxpy} package~\cite{diamond2016cvxpy} utilizing \texttt{MOSEK}~\cite{mosek}  to solve the SDP programs.   
\subsection*{Open-loop unstable system: Double integrator}
We consider an open-loop unstable double integrator system having  dynamics  described as in~\eqref{eq:system} with
\begin{equation}\label{eq:example}
\begin{bmatrix}
    x_{1}(k+1) \\ x_{2}(k+1)
\end{bmatrix} = 
\underbrace{\begin{bmatrix}
    1 & 1\\
    0 & 1
    \end{bmatrix}}_{A} \begin{bmatrix}
        x_{1}(k)\\x_{2}(k)
    \end{bmatrix}+ 
    \underbrace{\begin{bmatrix}
   0\\1
    \end{bmatrix}}_{B} u(k) + w(k).
\end{equation}
Note that the system matrices $(A,B)$ are \emph{unknown}, but they are only used to gather the data. A single state-input trajectory of $T=20$ samples is gathered by exciting the system~\eqref{eq:example} with inputs uniformly distributed in $ [-2, 2]$, see Fig.~\ref{fig:data}. The data satisfies the persistency of excitation rank conditions given in Proposition~\ref{prop:rich_data}. 
The disturbance $w$ acting on the systems is assumed to take values in the bound $[ -0.1,\ 0.1]$, \emph{i.e.}, $D=10$ according to the set definition $\mathcal{W}$ in \eqref{eq:constr}. The state constraints are $(x_1, x_2) \in [-2, 2] \times [-2, \ 2]$ and the input constraints are $u \in [-2, 2]$.   

\begin{figure}[h!]
\centering
    \includegraphics[width=\columnwidth]{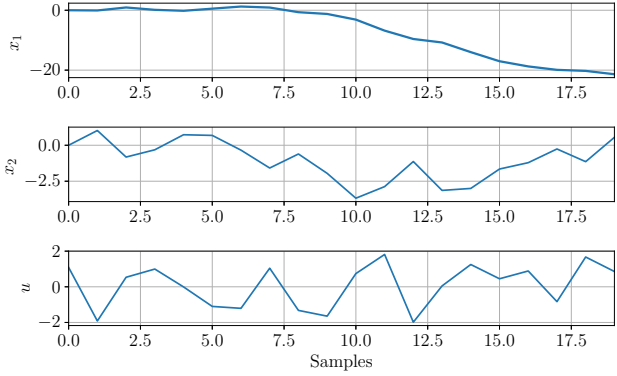}
    \caption{State-input data gathered from the double integrator system. }
    \label{fig:data}
\end{figure}

\subsubsection{Comparison between  data-driven approaches and a model-based method}

In this subsection, we compare the proposed data-driven algorithms with a model-based approach~\cite{gkf17}. In the model-based approach, exact values of the system matrices $(A,B)$ are assumed to be known. 
The complexity of the RCI sets is selected as $n_p = 3$ by choosing matrix $P$ as follows 
\begin{equation*}
    P = \begin{bmatrix}
        10 & 10 \\ 10 & 0 \\1 & 11
    \end{bmatrix}
\end{equation*}

\begin{figure}[t!]
	\centering
\begin{subfigure}[h]{\columnwidth}
\centering
	\includegraphics[width=0.75\columnwidth]{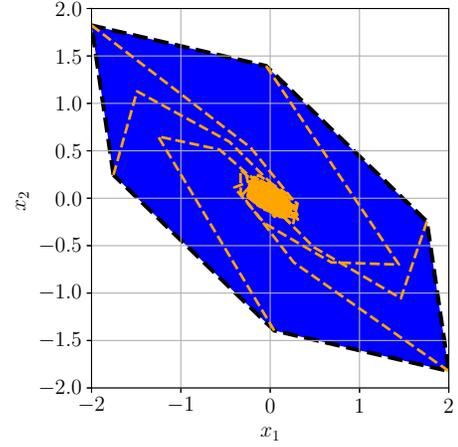}
	\caption{Data-driven approach: $1$-step algorithm~\eqref{eq:one-step SDP} }
	\label{fig:rci_data}
\end{subfigure}
\begin{subfigure}[h]{\columnwidth}
\centering
	\includegraphics[width=0.75\columnwidth]{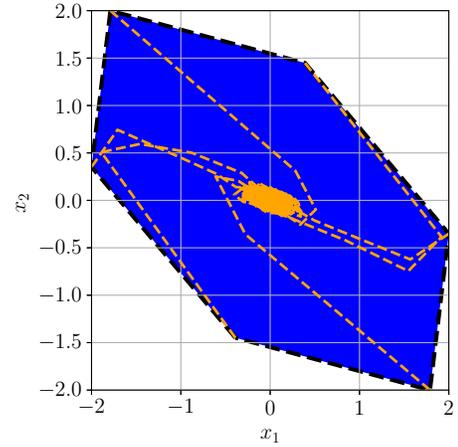}
	\caption{Data-driven iterative approach with dilated LMI~\eqref{eq:iterative SDP}}
	\label{fig:rci_data_dil}
\end{subfigure}
\begin{subfigure}[h]{\columnwidth}
\centering
	\includegraphics[width=0.75\columnwidth]{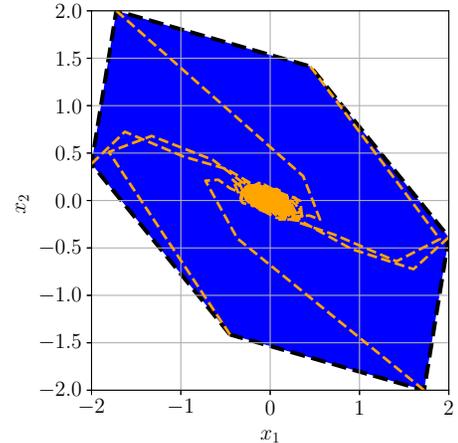}
	\caption{Model-based approach with dilated LMIs~\cite{gkf17}}
	\label{fig:rci_data_model}
\end{subfigure}
\caption{Direct data-driven approach \emph{vs} model-based approach: RCI sets (blue) with closed-loop state trajectories (dashed yellow).}
\label{fig:rcisets}
\end{figure}

\begin{figure}[h!]
\centering
    \includegraphics[width=\columnwidth]{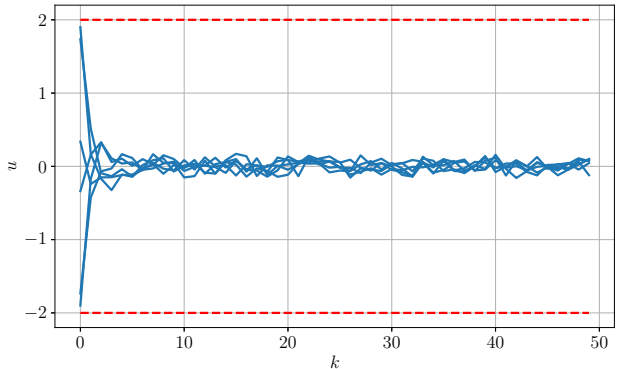}
    \caption{Closed-loop simulation: Control input $u = \K x$ trajectories for the computed state-feedback gain (blue) and input constraints (dashed-red).}
    \label{fig:input constrains}
\end{figure}

The RCI sets and the associated state-feedback control laws are computed by running one-step \textbf{Algorithm 1} solving~\eqref{eq:one-step SDP} and \textbf{Algorithm 2} solving~\eqref{eq:iterative SDP} iteratively for $5$ iterations with dilated LMI conditions. We also compute the RCI set and control law based on dilated LMI conditions given in the model-based method~\cite{gkf17}. The resulting RCI sets matrices and the state-feedback gains   are obtained as follows:
\begin{align*}
\left [ \begin{array}{cc} 
\W\\\hline
\K
\end{array} \right ] &=\left [ \begin{array}{cc} 
\!\!\;\;17.54  &  \;\;-2.46\!\!\\
\!\!-2.46  &     \;\;15.77\!\!\\\hline
\!\!   -0.71  & -1.45\!\!\\
\end{array} \right ],   \ \text{(Data-driven: $1$-step)} \\
\left [ \begin{array}{cc} 
\W\\\hline
\K
\end{array} \right ] &=\left [ \begin{array}{cc} 
\!\!\;\;20.00  &  \;\;2.11\!\!\\
\!\!-3.51  &     \;\;16.48\!\!\\\hline
\!\!   -0.38  & -1.21\!\!\\
\end{array} \right ],   \ \text{(Data-driven: iterative)} \\
\left [ \begin{array}{cc} 
\W\\\hline
\K
\end{array} \right ] &=\left [ \begin{array}{cc} 
\!\!\;\;20.00  &  \;\;2.77\!\!\\
\!\!-3.87  &     \;\;16.13\!\!\\\hline
\!\!   -0.41  & -1.18\!\!\\
\end{array} \right ] \ \text{(Model-based)}
\end{align*}

\begin{table}[t!]
\centering
\caption{Volume of the RCI set obtained with data-driven (DD) algorithms and model-based (MB) approach.} \label{tab:volume_methods}
\begin{tabular}{|c||c|c|c|}
\hline 
 & DD: $1$-step & DD: iterative & MB~\cite{gkf17} \\
 \hline
 Volume of $\mathcal{C}$ & 8.31& 9.86& 9.50 \\
 \hline
 \end{tabular}
\end{table}

The obtained RCI sets are depicted in Fig.~\ref{fig:rcisets}. It can be observed that the proposed direct data-driven approach is able to generate RCI sets which are of comparable volume to those obtained with the model-based method. The main advantage is that explicit knowledge of model matrices $(A,B)$ is not required, thus avoiding an additional identification step. 
The corresponding volumes of the RCI sets are reported in Table~\ref{tab:volume_methods}, which shows  that iterative \textbf{Algorithm~2}  with dilated LMI conditions generates relatively larger size RCI sets than those computed with the one-step \textbf{Algorithm~1}, which indicate that \textbf{Algorithm~2} is indeed less conservative for this example.

Furthermore, Fig.~\ref{fig:rcisets} also shows closed-loop state trajectories starting from each vertex of the RCI set. These trajectories are obtained by simulating the true system in  closed-loop with the state-feedback controller $u=\K x$. During the closed-loop simulation, a random disturbance uniformly distributed in the interval $[-0.1, 0.1 ]$  is acting on the system at each time instance. The figure shows that the approach guarantees robust invariance in the presence of a bounded but unknown disturbance while respecting the state-constraints.  The corresponding input trajectories computed with the iterative data-driven algorithm are shown in Fig.~\ref{fig:input constrains}. The figure  shows that the input constraints are also satisfied. 


\subsubsection{RCI sets with different complexities}

In this subsection, we analyse the effect of choosing different $P$ matrices corresponding to different complexities of polytope. The  RCI set and the associated state-feedback gain matrices are computed  running \textbf{Algorithm 1} and the computed matrices are as follows,
\begin{align}\label{rcisets_description}
&P_2=\begin{bmatrix}
1   &  0\\
0   &  1
\end{bmatrix},\quad \quad  \left [ \begin{array}{cc} 
\W_2\\\hline
\K_2
\end{array} \right ]=\left [ \begin{array}{cc} 
\!\!\;\;1.33  & -0.67\!\!\\
\!\!\;\;-0.67  &  \;\;1.17\!\!\\\hline
\!\!   -0.87  & -1.89\!\!\\
\end{array} \right ], \nonumber \\
&P_3\!=\!\begin{bmatrix}
 \!\;\;20\!\!  & \!\! \;\;20\!\\
 \! -20 \!\!&  \!\!  \;\;0\!\\
 \!\;\;0 \!\! &\!\! -25\!
\end{bmatrix},\left [ \begin{array}{cc} 
\W_3\\\hline
\K_3
\end{array} \right ]\!=\!\left [ \begin{array}{cc} 
35.32 & -5.85\\
   -5.85 &  \;\;33.11\\\hline
   -0.64 &  -1.52\\
\end{array} \right ], \nonumber \\
&P_4\!=\!\begin{bmatrix}
\!-18 \!\!&\!\!  -55\!\\
\!\;\; 18\!\! &\!\!  \;\; \,55\!\\
\!\;\;55 \!\!& \!\! -18\!\\
\!\;\;55 \!\!&\!\!   \;\;18\!\\
\end{bmatrix},\left [ \begin{array}{cc} 
\W_4\\\hline
\K_4
\end{array} \right ]\!\!=\!\!\left [ \begin{array}{cc} 
\!\!\;\;110.00  & -29.08\!\!\\
\!\!\;-29.08  & \;\; 77.01\!\!\\\hline
\!\!-0.56 &  -1.52\!\!
\end{array} \right ].
\end{align}

\begin{figure}[t!]
    \centering
    \includegraphics[width=0.75\columnwidth]{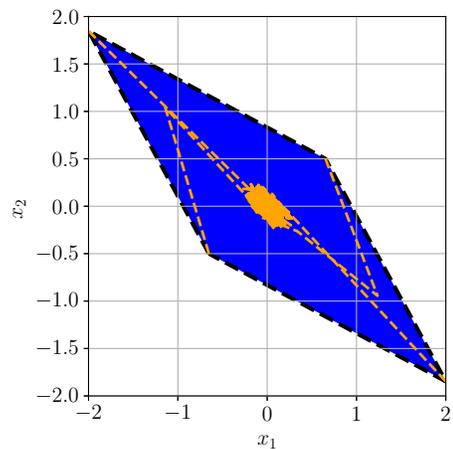}
    \includegraphics[width=0.75\columnwidth]{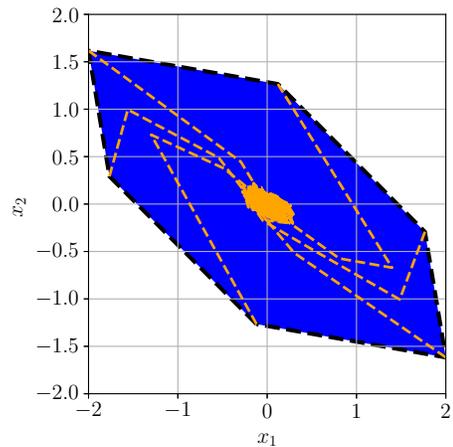}
    \includegraphics[width=0.75\columnwidth]{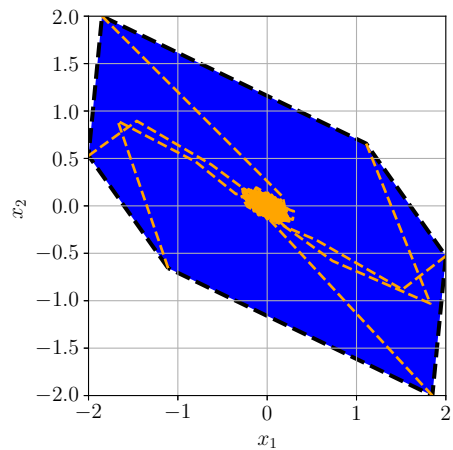}
    \caption{Maximum volume RCI sets $\mathcal{C}$ with different complexities: $n_p=2$ (top panel), $n_p=3$ (middle panel), $n_p=4$ (bottom panel).}
    \label{fig:rciset_P}
\end{figure}

\begin{table}[b!]
 \caption{Volumes of the RCI sets with different complexities} \label{tab:volume_complexity}
 \centering
\begin{tabular}{|c||c|c|c|}
\hline 
Complexity & $n_p=2$ & $n_p =3$ & $n_p=4$ \\
 \hline
 Volume of $\mathcal{C}$ & 4.19& 7.04& 8.29 \\
 \hline
 \end{tabular}
\end{table}

The subscripts in (\ref{rcisets_description}) indicate the set complexity $n_p$. In Fig.~\ref{fig:rciset_P} maximum volume RCI set with complexities $n_p=2,3,4$ are plotted along with the  closed-loop state trajectories obtained by simulating the true system from different initial conditions and randomly varying the disturbance within the chosen bounds. The corresponding volume of the RCI sets are reported in Table~\ref{tab:volume_complexity}.

As expected, it can be observed that as $n_p$ increases, size of the RCI set increases, thus $n_p$ can be used as an additional tuning parameter to obtain an invariant set with larger volume.


\section{CONCLUSIONS}
We proposed a direct data-driven approach  to compute a full complexity polytopic RCI set and an associated linear state-feedback control law.  In the proposed algorithm neither the model of the system is required to be known nor any identification step is necessary. The algorithm is robust w.r.t. a set of all feasible models compatible with the available state-input data and satisfying the disturbance bounds.  The  proposed direct data-driven approach is able to generate RCI sets with sizes that are comparable to that of an approach in which exact system knowledge is assumed. As a future work, the proposed approach can be extended to generate RCI sets and controllers for a more general class of systems, \emph{e.g.}, linear parameter-varying and non-linear systems.


\bibliographystyle{plain}
\bibliography{cdc23_ref}

\end{document}